\documentclass[11pt]{article}
\usepackage{geometry}

\usepackage[english]{babel}
\usepackage{amssymb,amsmath}
\usepackage{array}
\usepackage{amsthm}
\usepackage[symbol]{footmisc}
\usepackage{graphicx}
\usepackage[shortlabels]{enumitem}
\usepackage[dvipsnames]{xcolor} 
\usepackage{algorithm}
\usepackage{algpseudocodex}
\usepackage{authblk}
\usepackage{mathtools}
\usepackage{bbm}

\newcommand{\brak}[1]{\left\langle #1\right\rangle}

\newcommand{\infnorm}[1]{\Vert #1\Vert_\infty}
\renewcommand{\Bar}[1]{\overline{#1}}

\DeclareMathOperator{\polylog}{polylog}

\newcommand{\Number}[1]{\mathbb{#1}}

\newcommand{\Text}[1]{\text{#1}}

\newcommand{\R}{\Number{R}}

\newcommand{\sS}{\Number{S}}

\newcommand{\dint}{d_1}
\newcommand{\dmax}{d_{\infty}}
\newcommand{\wt}{\Text{c}}

\DeclareMathOperator{\PHT}{\mathrm{PHT}}
\DeclareMathOperator{\Dgm}{\mathrm{Dgm}}

\newtheorem*{theorem*}{Theorem}
\newtheorem{theorem}{Theorem}[section]

\newtheorem{corollary}[theorem]{Corollary}
\newtheorem{lemma}[theorem]{Lemma}

\theoremstyle{definition}

\newtheorem{remark}[theorem]{Remark}

\usepackage{titlesec}
\titlespacing*{\paragraph}{0pt}{0.4\baselineskip}{0.5em}

\title{\LARGE Computing the Bottleneck Distance between Persistent Homology Transforms}
\author[1]{\large Michael Kerber}
\author[1,2]{\large Elena Xinyi Wang}

\affil[1]{\footnotesize Institute of Geometry, Graz University of Technology}
\affil[2]{\footnotesize Department of Informatics, University of Fribourg}

\begin{document}
\date{}
\maketitle

\begin{abstract}
The Persistent Homology Transform (PHT) summarizes a shape in $\R^m$ by collecting persistence diagrams obtained from linear height filtrations in all directions on $\mathbb{S}^{m-1}$. 
It enjoys strong theoretical guarantees, including continuity, stability, and injectivity on broad classes of shapes. 
A natural way to compare two PHTs is to use the bottleneck distance between their diagrams as the direction varies. 
Prior work has either compared PHTs by sampling directions or, in 2D, computed the exact \emph{integral} of bottleneck distance over all angles via a kinetic data structure. 
We improve the integral objective to $\tilde O(n^5)$ in place of the earlier $\tilde O(n^6)$ bound, where $n$ denotes the number of simplices. 
For the \emph{max} objective, we give an $\tilde O(n^3)$ expected-time algorithm in $\R^2$ and an $\tilde O(n^5)$ expected-time algorithm in $\R^3$.
\end{abstract}

\paragraph*{Acknowledgements.} This research has been supported by the Austrian Science Fund (FWF) grant P 33765-N.

\section{Introduction}
Persistent homology captures the multiscale evolution of homology classes along a filtration induced by a real-valued function on a space. 
It has proved useful across fields such as materials science~\cite{edelsbrunner2024mergetree, hiraoka2016hierarchical}, computational biology~\cite{giusti2015clique, reimann2017cliques}, as well as machine learning~\cite{carriere2020perslay, Hensel2021Survey}, providing robust descriptors for high-dimensional data.
A persistence diagram is a concise multiset summary of persistent homology, recording the birth--death times of features. 
Similarity between objects is then assessed by comparing their diagrams, most commonly via the bottleneck distance. 
Computationally, this reduces to a minimum-bottleneck perfect matching on a bipartite graph with diagonal copies, and algorithms for this problem have been studied extensively in both theory and practice~\cite{Cabello2024Matching, Dey2018Bottleneck,Efrat2001,HopcroftKarp1973,KerberGeometryHelps2017}.

Since a single diagram depends on the chosen filter and may miss geometry, the Persistent Homology Transform (PHT) collects diagrams from all directions to retain fuller information~\cite{TurnerPHT2014}.
It assigns to a shape in $\R^m$ the collection of persistence diagrams obtained from linear height filtrations over all directions on $\sS^{m-1}$. 
It possesses many desirable theoretical properties such as continuity, stability, and injectivity~\cite{CurryPHT2018,GhristPHT2018,TurnerPHT2014}.
These properties make the PHT an effective topological descriptor for applications.
However, unlike other directional-transform--based signatures such as the Euler Characteristic Transform, the PHT lacks efficient comparison algorithms, limiting its practical use.
A natural way to compare PHTs is to sample directions and compute the bottleneck distance only at those angles.
This inevitably leads to information loss. 
To integrate the bottleneck distance over the parameter domain (e.g., $\sS^{m-1}$), prior work by Munch, Wang, and Wenk introduced a kinetic ``hourglass'' approach that maintains matchings as the parameter moves.
It computes this integral with an $\tilde O(n^6)$ bound for data in $\R^2$~\cite{munch2025kinetic}. 
In this work, we focus on the max variant, which we define as the supremum of the bottleneck distance over all directions.

\paragraph*{Our Contributions:}
We provide an $\tilde O(n^3)$ expected-time algorithm to compute the max bottleneck distance between PHTs of shapes in $\R^2$ and an $\tilde O(n^5)$ expected-time algorithm in $\R^3$,
where $\tilde{O}$ means that we ignore logarithmic factors.
Additionally, we improve the computation of the integral distance from the previous $\tilde O(n^6)$ bound to $\tilde O(n^5)$ in $\R^2$.

For the max bottleneck distance, we use techniques from computational geometry and topology and inspired by the strategy of Bjerkevik and Kerber for computing the matching distance for two-parameter modules~\cite{Bjerkevik2023}.
At a high level, the solution for 2D and 3D are similar.
We first reformulate the problem as a decision problem: given a threshold $\lambda$, we determine whether the max bottleneck distance is at most $\lambda$. 
To solve this, we sweep the direction parameter around the sphere while maintaining the bottleneck matching between the two persistence diagrams. 
The key insight is that this matching changes only at discrete ``events'': points where persistence diagrams change combinatorially or where edge weights cross the threshold $\lambda$. 
At each event, we update the matching using a single augmenting path search in the geometric bipartite graph, exploiting the fact that only local changes occur.

We then identify that the maximum bottleneck distance must occur at one of a finite set of candidate values, corresponding to critical points and intersections of the difference functions between simplices. 
While there are $O(n^4)$ candidates in 2D and $O(n^6)$ in 3D, we avoid exhaustive enumeration through a randomized incremental construction. 
We maintain a band containing the true maximum and iteratively refine it by processing simplices in random order, testing whether their associated candidates lie within the current band. 
The randomization ensures that we expect only $O(\log n)$ refinements, reducing the expected complexity to $\tilde{O}(n^3)$ in 2D and $\tilde{O}(n^5)$ in 3D.

\paragraph*{Outline:}
We begin in Section~\ref{Sec:Background} with a brief overview of persistent homology, persistence diagrams, and the PHT, and we formalize the problem.
Section~\ref{Sec:2D} presents our $\R^2$ algorithm: first at a high level, then with full technical details. 
The improvement for the integral distance is shown in Section~\ref{Sec:2Dint}.
Section~\ref{Sec:3D} extends the approach to $\R^3$ before concluding in Section~\ref{Sec:Conclusion}.

\section{Preliminary}
\label{Sec:Background}
\subsection{Persistent Homology}
Persistent homology is a multi-scale summary of the connectivity of objects in a nested sequence of subspaces; see \cite{DeyWang2017,EdelsHarer2010} for an introduction. 
    For the purposes of this section, we define a 
    \emph{persistence diagram} to be a finite collection of points $\{(b_i,d_i)\}_i$ with $d_i \geq b_i$ for all $i$. 
    
    Given two persistence diagrams $X$ and $Y$, a partial matching is a bijection $\eta:X' \to Y'$ on a subset of the points $X' \subseteq X$ and $Y' \subseteq Y$; we denote this by $\eta: X \rightleftharpoons Y$. 
    The cost of a partial matching is the maximum over the $L_\infty$-norms of all pairs of matched points and the distance between the unmatched points to the diagonal:
    \begin{equation}
    \label{eqt:bottleneck}
        c(\eta) = \max \left( 
        \{ \|x-\eta(x)\|_\infty \mid x \in X'\} 
        \cup
        \{ \tfrac{1}{2}|z_2-z_1|  \mid 
        (z_1,z_2) \in (X \setminus X') \cup (Y \setminus Y') \}
        \right)
    \end{equation}
    and 
    the bottleneck distance is defined as
    $
    d_B(X, Y) = \inf_{\eta:X\rightleftharpoons Y} c(\eta)
    $.

    We can reduce finding the bottleneck distance between persistence diagrams to the problem of finding the bottleneck cost of a bipartite graph.
    Let $X$ and $Y$ be two persistence diagrams given as finite lists of off-diagonal points. 
    For any off-diagonal point $z = (z_1, z_2)$, the orthogonal projection to the diagonal is 
    $z' = ((z_1+z_2)/2, (z_1+z_2)/2)$. 
    Let $\Bar{X}$ (resp.~$\Bar{Y}$) be the set of orthogonal projections of the points in $X$ (resp.~$Y$). 
    Set $U = X\sqcup \Bar{Y}$ and $V = Y\sqcup \Bar{X}$.
    We define the complete bipartite graph $G = (U\sqcup V, U\times V, \wt)$, where for $u\in U$ and $v\in V$, the weight function $\wt$ is given by
    \[
    \wt(uv) = 
    \begin{cases}
        \infnorm{u - v} &\text{if $u\in X$ or $v\in Y$}\\
        0 &\text{if $u\in \Bar{X}$ and $v\in \Bar{Y}$}.
    \end{cases}
    \]
    An example of the bipartite graph construction is shown in Figure \ref{fig:persBottleneck}.
    This graph can be used to compute the bottleneck distance of the input diagrams because of the following lemma.
    \begin{lemma}[Reduction Lemma \cite{EdelsHarer2010}]
    \label{lem:reduction}
        For the above construction of $G$,  $d_B(G) = d_B(X, Y)$.
    \end{lemma}

    \begin{figure}
    \centering
    \includegraphics[width=0.6\linewidth]{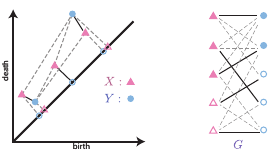}
    \caption{Construction of the bipartite graph $G$ based on the persistence diagrams $X$ and $Y$. 
    }
    \label{fig:persBottleneck}
\end{figure}

\subsection{Persistent Homology Transform}
Let $K$ be a finite simplicial complex embedded in $\R^m$. 
For any direction $\omega \in \sS^{m-1}$, the height function $h_\omega: |K| \to \R$ is defined by $h_\omega(x) = \brak{x, \omega}$. 
The lower-star filtration induced on the abstract complex $K$ is given by:
\[
h_\omega(\sigma) = \max\{ h_\omega(v) \mid v \in \sigma \},
\]
and the sublevelset at height $a$ is the full subcomplex induced by vertices $\{v \in K \mid h_\omega(v) \leq a\}$. Filtering $K$ by $h_\omega$ produces a persistence diagram $\Dgm_k(h^K_\omega)$.
We focus on $k = 0, 1$ for $m = 2$, and $k = 0, 1, 2$ for $m = 3$.

The \emph{persistent homology transform (PHT)} is defined as:
\[
\PHT(K): \sS^{m-1} \to \mathcal{D}, \quad \omega \mapsto \Dgm(h_\omega^K),
\]
where $\mathcal{D}$ is the space of persistence diagrams. The PHT is injective \cite{TurnerPHT2014, CurryPHT2018}, making it a faithful shape representation.

\subsection{Maximum Bottleneck Distance}
In this work, we are interested in comparing two finite embedded simplicial complexes, $K, K'\in\R^m$, $m = 2$ or $3$, by comparing their PHTs. 
We define the \emph{maximum bottleneck distance} between $K$ and $K'$ to be:
\[\dmax(\text{PHT}(K), \text{PHT}(K')) = \max_{\omega\in\sS^{m-1}}d_B\left(\mathrm{Dgm}(h_\omega^{K}), \mathrm{Dgm}(h_\omega^{K'})\right).\]

This distance is a variant of a similar integral distance that has been studied previously by Munch, Wang, and Wenk~\cite{munch2025kinetic}:
\[\dint(\text{PHT}(K), \text{PHT}(K')) = \int_{\sS^{m-1}}d_B\left(\mathrm{Dgm}(h_\omega^{K}), \mathrm{Dgm}(h_\omega^{K'})\right)d\omega.\]

Both $d_\infty$ and $d_1$ inherit stability from the bottleneck distance.
However, they differ from a computational perspective.
Computing $d_1$ requires full knowledge of the whole domain, which is inevitably expensive. 
Our $d_\infty$ metric seeks the maximum value, which allows us to reformulate the problem as finding the supremum over a finite set of candidate values rather than tracking the bottleneck distance continuously.

A key property that makes our algorithmic approach feasible is that the bottleneck distance function is Lipschitz continuous, which is a direct consequence of the stability of the bottleneck distance:
\begin{lemma}
\label{lem:continuity}
Let $H(\omega):=d_B\bigl(\mathrm{Dgm}(h^K_\omega),\mathrm{Dgm}(h^{K'}_\omega)\bigr)$ denote the bottleneck distance between persistence diagrams. 
Then the function $H: \mathbb{S}^{m-1} \to \mathbb{R}$ is Lipschitz continuous with
constant $R_K + R_{K'}$, where $R_K = \max_{v \in K^{(0)}} \|v\|$ and
$R_{K'} = \max_{v \in (K')^{(0)}} \|v\|$.
\end{lemma}

\begin{proof}
For any $\omega, \omega' \in S^{m-1}$ and any vertex $v \in K^{(0)}$, we have
\[
|h^K_\omega(v) - h^K_{\omega'}(v)| = |\langle v, \omega \rangle - \langle v, \omega' \rangle| = |\langle v, \omega - \omega' \rangle| \leq \|v\| \cdot \|\omega - \omega'\|.
\]
Since $K$ is a finite simplicial complex, let $R_K = \max_{v \in K^{(0)}} \|v\|$ denote the maximum norm of any vertex. Then
\[
\|h^K_\omega - h^K_{\omega'}\|_\infty \leq R_K \cdot \|\omega - \omega'\|.
\]
By the stability theorem for persistence diagrams, which states that $d_B(\Dgm(f), \Dgm(g)) \leq \|f - g\|_\infty$ for tame functions $f$ and $g$, we obtain
\[
d_B\bigl(\Dgm(h^K_\omega), \Dgm(h^K_{\omega'})\bigr) \leq R_K \cdot \|\omega - \omega'\|.
\]
Thus $\omega \mapsto \Dgm(h^K_\omega)$ is Lipschitz continuous, and by the same argument, $\omega \mapsto \Dgm(h^{K'}_\omega)$ is Lipschitz continuous with constant $R_{K'} = \max_{v \in (K')^{(0)}} \|v\|$.

Now, the reverse triangle inequality for the bottleneck metric gives
\[
|d_B(A, B) - d_B(C, D)| \leq d_B(A, C) + d_B(B, D)
\]
for any persistence diagrams $A, B, C, D$. Applying this with $A = \Dgm(h^K_\omega)$, $B = \Dgm(h^{K'}_\omega)$, $C = \Dgm(h^K_{\omega'})$, and $D = \Dgm(h^{K'}_{\omega'})$, we obtain
\begin{align*}
|H(\omega) - H(\omega')| &\leq d_B\bigl(\Dgm(h^K_\omega), \Dgm(h^K_{\omega'})\bigr) + d_B\bigl(\Dgm(h^{K'}_\omega), \Dgm(h^{K'}_{\omega'})\bigr) \\
&\leq (R_K + R_{K'}) \|\omega - \omega'\|.
\end{align*}
Therefore $H$ is Lipschitz continuous.
\end{proof}

\section{Computing $\dmax$ in $\R^2$}
\label{Sec:2D}
Our setup in this section is as follows: we are given two geometric simplicial complexes $K$, $K'$ in $\R^2$ with at most $n$ simplices each.
The goal is to compute $\dmax$: we first introduce the necessary notation, then give a high-level overview, and finally present the technical details.

\subsection{Events}
For each simplex $\alpha \in K \cup K'$, we define an insertion curve $I_\alpha(\omega)$. 
These curves record the height at which each simplex enters the sublevel-set filtration induced by direction $\omega \in S^1$. 
Explicitly, if the $k$-simplex $\alpha$ has vertices $\{v_0, \ldots, v_k\} \subseteq V(K)$, its insertion curve is
\[
I_\alpha(\omega) := \max_{0 \leq i \leq k} h_\omega(v_i), \quad h_\omega(x) = \langle x, \omega \rangle.
\]
Each $I_\alpha$ is a piecewise-trigonometric function on $S^1$.

There are at most $2n$ such curves. We now define \emph{difference curves} that capture both matched and unmatched contributions to the bottleneck distance. For distinct simplices $\alpha, \sigma$:
\[
\Delta_{\alpha,\sigma}(\omega) := \begin{cases}
|I_\alpha(\omega) - I_\sigma(\omega)| & \text{if } \alpha, \sigma \text{ are from different complexes}  \\
\frac{1}{2}|I_\alpha(\omega) - I_\sigma(\omega)| & \text{if } \alpha, \sigma \text{ are from the same complex}
\end{cases}
\]
These are real-valued functions on $S^1$; geometrically, one can picture their graphs as
curves in $S^1 \times \R_{\geq 0}$.
The factor of $\frac{1}{2}$ corresponds to the contribution of unmatched points to the bottleneck distance, see Equation~\ref{eqt:bottleneck}.
There are $O(n^2)$ such curves for $|K|, |K'| = O(n)$. 
We call the vertex $v_i$ achieving $\max_{0 \leq j \leq k}\langle v_j, \omega \rangle$ the \emph{active vertex} of $\alpha$ at $\omega$. 
As $\omega$ varies, the active vertex changes at discrete directions where two vertices achieve the same height. 
Within any interval where the active vertices remain fixed, the difference curve simplifies to either $|\langle u - w, \omega \rangle|$ (for different complexes) or $\frac{1}{2}|\langle u - w, \omega \rangle|$ (for same complex), which is a simple trigonometric function. 
See Figure~\ref{fig:sector} for an example. 
The insertion and difference curves are both well-behaved. See Figure~\ref{fig:2d-events} for an example.

\begin{figure}
    \centering
    \includegraphics[width=0.8\linewidth]{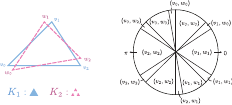}
    \caption{On the left are two simplices with vertices $v_{0, 1, 2}$ and $w_{0, 1, 2}$ respectively. On the right is $\sS^1$ partitioned into regions with the same active vertex pairs $(v_i, w_j)$.}
    \label{fig:sector}
\end{figure}

\begin{lemma}
\label{lem:event-bound}
Each difference curve $\Delta_{\alpha,\sigma}$ crosses any horizontal line $y = \lambda$ at most $O(1)$ times as $\omega$ varies over $[0,2\pi]$. 
Since there are $O(n^2)$ difference curves in total, this yields $O(n^2)$ crossings with any horizontal line. 
In particular, there are $O(n^2)$ insertion-curve intersections, as these correspond to difference curves crossing the line $y=0$.
\end{lemma}

The bound on insertion-curve intersections follows immediately since we have $O(n)$ curves total, giving at most $\binom{O(n)}{2} = O(n^2)$ pairwise intersections. 
For threshold crossings, each difference curve $\Delta_{\alpha,\sigma}$ is piecewise trigonometric with $O(1)$ pieces (as simplices in $\mathbb{R}^2$ have at most 3 vertices), and crosses any threshold $\lambda$ at most $O(1)$ times. 
With $O(n^2)$ pairs total, this yields $O(n^2)$ threshold crossings. 
The detailed geometric analysis is in Appendix~\ref{appendix:event-bound-proof}.

In each direction $\omega$, the bottleneck distance is realized by the value of a difference curve $\Delta_{\alpha,\sigma}(\omega)$. 
When this maximum is unique in a neighborhood of $\omega$, the continuity of the bottleneck function $H(\omega) := d_B\left(\text{Dgm}(h^K_\omega), \text{Dgm}(h^{K'}_\omega)\right)$ ensures that the same curve or persistence value realizes the bottleneck distance in an $\varepsilon$-neighborhood.
We call the points in $S^1$ where the realizing curve or value changes locally \emph{events}, and the values $H$ attains at events are the \emph{candidate values}.

\begin{theorem}[Events for $d_\infty$ in $\mathbb{R}^2$]\label{thm:R2events}
The direction $\omega^*$ attaining $d_\infty$ occurs at either
\begin{enumerate}[noitemsep, topsep=0pt]
\item a local maximum of a difference curve $\Delta_{\alpha,\sigma}$, or
\item an intersection of two difference curves. See Figure~\ref{fig:2d-events} for an illustration.
\end{enumerate}
\end{theorem}

\begin{proof}
Let $\omega^*$ be the direction that attains $d_\infty$. 
By Lemma~\ref{lem:continuity}, if the bottleneck distance at $\omega^*$ is realized by a unique difference curve in a neighborhood, then $\omega^*$ must be a local maximum of that difference curve, giving case (1). 
This includes both matched pairs (when $\alpha \in K, \beta \in K'$) and unmatched point persistence (when $\alpha, \beta$ are from the same complex).
Otherwise, at least two curves achieve the same maximum value at $\omega^*$, giving case (2). 
Note that intersections of insertion curves (where the persistence diagram changes combinatorially) form a special case of (2), as when two insertion curves $I_\alpha, I_\beta$ from the same complex coincide, all difference curves of the form $\Delta_{\alpha,\sigma}$ and $\Delta_{\beta,\sigma}$ are equal.
\end{proof}

\begin{remark}
When we refer to ``intersections'' of difference curves, we mean isolated transverse crossings, not intervals of coincidence.
Endpoints of overlaps also count as intersections, since the realizing curve may change at those boundaries.
\end{remark}

\begin{figure}[t]
    \centering
    \begin{minipage}[b]{0.32\textwidth}
        \centering
        \includegraphics[width=\textwidth]{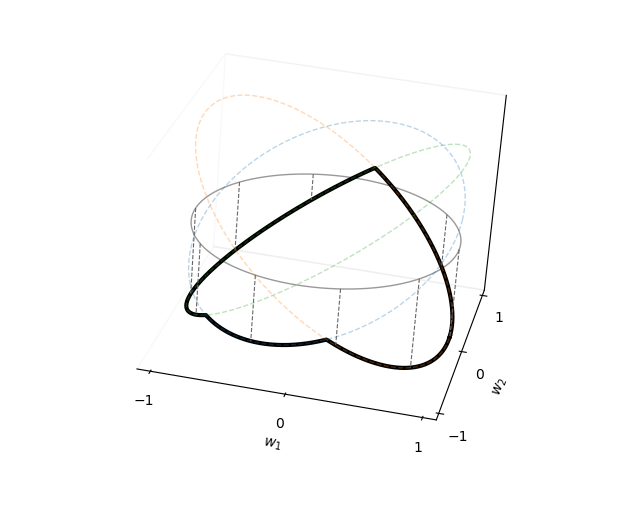}
        (a)
    \end{minipage}
    \hfill
    \begin{minipage}[b]{0.32\textwidth}
        \centering
        \includegraphics[width=\textwidth]{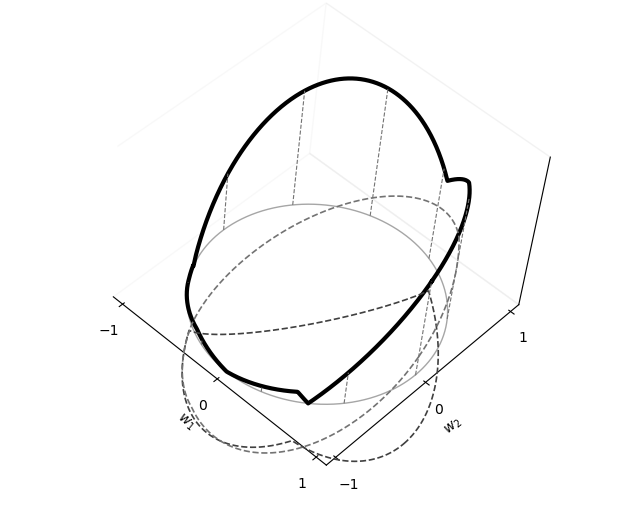}
        (b)
    \end{minipage}
    \hfill
    \begin{minipage}[b]{0.32\textwidth}
        \centering
        \includegraphics[width=\textwidth]{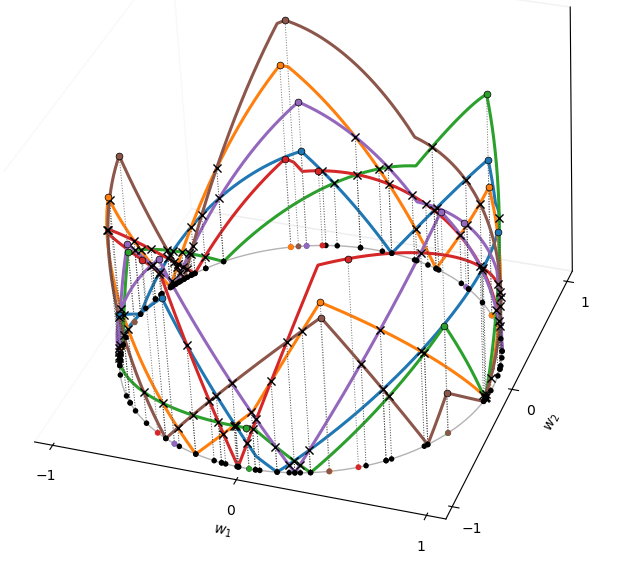}
        (c)
    \end{minipage}
    \caption{(a)~An insertion curve $I_\alpha(\omega)$.
             (b)~A difference curve $\Delta_{\alpha,\sigma}(\omega)$.
             (c)~Arrangement of four difference curves; local maxima and
             pairwise intersections mark the candidate values of
             Theorem~\ref{thm:R2events}.}
    \label{fig:2d-events}
\end{figure}

\begin{corollary}
\label{cor:candidate_count}
There are $O(n^4)$ candidate values where $d_\infty$ may be attained.
\end{corollary}
\begin{proof}
There are $O(n^2)$ difference curves: $O(n^2)$ from pairs $(\alpha,\sigma) \in K \times K'$, plus $O(n^2)$ from pairs within the same complex. 
Each has $O(1)$ local maxima (see Appendix~\ref{app:maxima} for the explicit bound), contributing $O(n^2)$ type-(1) candidates. 
For type-(2) candidates, each pair of difference curves intersects at $O(1)$ heights, giving $\binom{O(n^2)}{2} = O(n^4)$ intersections total. Thus the total number of candidates is $O(n^4)$.
\end{proof}

\begin{remark}
When simplices share vertices, their difference curves may partially coincide. 
For instance, if $\alpha, \beta \in K$ share vertex $v$ and $\sigma, \tau \in K'$ share vertex $w$, then $\Delta_{\alpha,\sigma}$ and $\Delta_{\beta,\tau}$ coincide in regions where $v$ and $w$ are simultaneously active. 
Despite these overlaps, each difference curve retains a unique identity at the combinatorial level through its defining simplex pair. 
In particular, our candidate enumeration remains valid: if the maximum bottleneck distance is achieved at direction $\omega^*$, then at least one combinatorially distinct difference curve must realize this maximum, and this curve's critical points and intersections are in our candidate set. 
The full technical detail is in Appendix~\ref{app:overlaps}.
\end{remark}

\subsection{Overview of the algorithm}
We describe the high-level structure of our algorithm next. On this level, our approach is similar to the algorithm of Bjerkevik and Kerber to compute the matching distance~\cite{Bjerkevik2023}. 
However, the approaches differ in the technical realization of the predicates which are described in the subsequent subsections.

\paragraph*{Decision procedure.} 
Given a threshold $\lambda \geq 0$, we decide whether $\dmax \leq \lambda$ by sweeping $\omega$ around $\mathbb{S}^1$ and maintaining a bottleneck matching. 
As $\omega$ rotates, we update the persistence diagrams at $O(n^2)$ insertion-curve intersections and perform single augmenting-path updates whenever an edge weight crosses the threshold. 
This decision procedure runs in $\tilde{O}(n^3)$ time.

\paragraph*{Basic algorithm.} 
By Theorem~\ref{thm:R2events}, $\dmax$ is attained at one of $O(n^4)$ candidate values corresponding to intersections of difference curves and their local maxima. 
These candidates can be enumerated and sorted in $O(n^4\log n)$ time. 
Combined with binary search using our decision procedure, this yields an $O(n^4\log n + n^3\polylog(n)) = \tilde O(n^4)$ algorithm.

\paragraph*{Improved algorithm via band refinement.} 
By Theorem~\ref{thm:R2events}, there are $O(n^4)$ total candidates where $\dmax$ could be attained.
To achieve $\tilde{O}(n^3)$ time, we avoid enumerating all candidates.

We call a candidate an $\alpha$-candidate if it arises from a difference curve $\Delta_{\alpha,\sigma}$ for some $\sigma$ (when $\alpha \in K$), from $\Delta_{\tau,\alpha}$ for some $\tau$ (when $\alpha \in K'$), or from $|I_\alpha(\omega) - I_\beta(\omega)|/2$ for some $\beta$ in the same complex as $\alpha$.

We maintain a height band $(\lambda_1, \lambda_2]$ satisfying $\dmax > \lambda_1$ and $\dmax \leq \lambda_2$. 
We process simplices from both complexes in random order. 
For each simplex $\alpha$, we apply:
\begin{enumerate}[noitemsep, topsep=0pt]
    \item Existence predicate: Tests if any $\alpha$-candidate lies in $(\lambda_1, \lambda_2]$ in $O(n^2 \polylog(n))$ time. If no candidate exists in the band, we skip to the next simplex.
    \item Band-refinement predicate: If the existence test is positive, we enumerate all $O(n^3)$ $\alpha$-candidates in the band and refine via binary search in $O(n^3\polylog(n))$ time.
\end{enumerate}

The resulting total running time is $O(n \cdot n^2\polylog(n) + C \cdot n^3\polylog(n))$, where $C$ is the number of band-refinement calls. 
As we will show, processing the simplices in a random order ensures $\mathbb{E}[C] = O(\log n)$, yielding expected time $\tilde{O}(n^3)$.

\subsection{Decision Problem}
We now describe how to decide whether $\dmax \le \lambda$ for a given threshold $\lambda \ge 0$. 
The key idea is to sweep the direction $\omega$ around $S^1$ and check if the bottleneck distance ever exceeds $\lambda$.

By Lemma~\ref{lem:reduction}, at each direction $\omega$, the bottleneck distance $d_B$ equals the cost of a minimum-bottleneck matching in a bipartite graph $G(\omega)$ constructed from the two persistence diagrams $X=\Dgm(h_\omega^{K})$ and $Y=\Dgm(h_\omega^{K'})$ with diagonal copies.
Define the subgraph $G_\lambda(\omega)$ to contain only those edges of $G(\omega)$ whose weights are $\le \lambda$. 
Then $d_B(X,Y)\le \lambda$ if and only if $G_\lambda(\omega)$ admits a perfect matching.

Our strategy is to maintain a perfect matching in $G_\lambda(\omega)$ as $\omega$ sweeps from $0$ to $2\pi$. 
If at any point no perfect matching exists, we conclude $\dmax > \lambda$; otherwise, if the sweep completes with a valid matching at every direction, then $\dmax \le \lambda$.

\paragraph*{Sweep algorithm.}
We initialize at $\omega=0$ by computing a perfect matching in $G_\lambda(0)$ using the Hopcroft-Karp algorithm in $O(n^{2.5})$~\cite{HopcroftKarp1973}. 
If no such matching exists, we conclude $\dmax>\lambda$.

As $\omega$ increases from $0$ to $2\pi$, the bipartite graph $G_\lambda(\omega)$ changes only at discrete events:
\begin{itemize}[noitemsep]
\item Diagram changes: At intersections of insertion curves, points move to or from the diagonal in the persistence diagrams, updating the bipartite graph structure.
\item Threshold crossings: When a difference curve $\Delta_{\alpha,\sigma}(\omega)$ crosses the threshold $\lambda$, the corresponding edge enters or leaves $G_\lambda(\omega)$.
\end{itemize}

\paragraph*{Maintaining the matching.} 
We maintain a current matching $M$ throughout the sweep. At each event:
\begin{itemize}[noitemsep, topsep=0pt]
\item If a new edge enters $G_\lambda(\omega)$ (weight decreases below $\lambda$), $M$ remains valid.
\item If an edge leaves $G_\lambda(\omega)$ (weight exceeds $\lambda$) and is not in $M$, then $M$ remains valid.
\item If an edge $(u,v) \in M$ becomes invalid, Hall's theorem guarantees that a new perfect matching exists if and only if there is an augmenting path from the newly exposed vertices. We search for such a path in $O(n \log n)$ time using the geometric structure~\cite{Efrat2001}. 
If found, we augment along this path to restore a perfect matching. 
If no augmenting path exists, then $G_\lambda(\omega)$ has no perfect matching, so $d_B > \lambda$ at this $\omega$, implying $\dmax > \lambda$.
\end{itemize}

At diagram change events (insertion curve intersections), we update the persistence diagrams using the vineyard algorithm~\cite{CohenSteiner2006Vineyard}. 
When two insertion curves $I_\alpha$ and $I_\beta$ from the same complex intersect, the relative order of exactly two simplices changes in the filtration. 
Following Bjerkevik and Kerber~\cite{Bjerkevik2023}, this corresponds to a single transposition in the barcode pairing, which can be updated in $O(n)$ time. 

If the sweep completes with a valid perfect matching at every $\omega$, then $\dmax \le \lambda$.

\begin{lemma}
\label{lem:decision-complexity}
The decision procedure runs in $\tilde{O}(n^3)$ time.
\end{lemma}
\begin{proof}
At each event, we either update the diagram structure in $O(n)$ time (vineyard update for insertion curve intersections) or search for a single augmenting path in $O(n\log n)$ time (for threshold crossings). 
With $O(n^2)$ total events from Lemma~\ref{lem:event-bound}, the total time is $\tilde{O}(n^3)$.
\end{proof}

\subsection{Candidate Pruning via Band Refinement}

Recall that we have $O(n^4)$ total candidates where $\dmax$ may be attained. To find $\dmax$ efficiently, we maintain a height band $(\lambda_1, \lambda_2]$ containing $\dmax$ and iteratively shrink it by processing simplices from both complexes in random order.

We distinguish between two types of intersections between difference curves. \emph{Intra-intersections} occur when two difference curves share a common simplex: either $\Delta_{\alpha,\sigma_1}$ and $\Delta_{\alpha,\sigma_2}$ sharing simplex $\alpha$, or $\Delta_{\beta_1,\tau}$ and $\Delta_{\beta_2,\tau}$ sharing simplex $\tau$. 
\emph{Inter-intersections} occur when $\Delta_{\alpha,\sigma}$ and $\Delta_{\beta,\tau}$ meet with all four simplices distinct: $\alpha \neq \beta$ and $\sigma \neq \tau$.

\paragraph*{Pre-refinement.}
Before processing individual simplices, we refine the band to exclude certain events from its interior: let $S$ be the set containing all local maxima of difference curves $\Delta_{\alpha,\sigma}$ for all pairs (including both different-complex and same-complex pairs), all intra-$\alpha$ intersections where $\Delta_{\alpha,\sigma_1}(\omega) = \Delta_{\alpha,\sigma_2}(\omega)$ for some $\omega \in \mathbb{S}^1$ and distinct $\sigma_1, \sigma_2$, and all intra-$\tau$ intersections where $\Delta_{\beta_1,\tau}(\omega) = \Delta_{\beta_2,\tau}(\omega)$ for distinct $\beta_1, \beta_2$.

Since there are $O(n^2)$ difference curves, each with $O(1)$ local maxima, we have $O(n^2)$ maxima total. For intra-intersections, each of the $O(n)$ simplices in $K$ contributes $\binom{O(n)}{2} = O(n^2)$ pairs of curves that could intersect, with $O(1)$ intersection points each, giving $O(n^3)$ intra-$\alpha$ intersections. Similarly for intra-$\tau$ intersections. Thus $|S| = O(n^3)$.

We binary search with the decision procedure to find the largest value in $S$ below $\dmax$ and the smallest at or above it, setting these as $\lambda_1$ and $\lambda_2$.
The resulting band $(\lambda_1, \lambda_2]$ contains no intra-events in its interior.
This takes $O(n^3 \polylog n)$ time.

\subsubsection{Predicate 1: Existence Test}

After pre-refinement, we process simplices to detect inter-intersections.

\begin{lemma}
\label{lem:existence_test}
For fixed $\alpha \in K$, determining whether any inter-intersection involving $\alpha$ lies in $(\lambda_1, \lambda_2]$ takes $O(n^2\log n)$ time.
\end{lemma}

\begin{proof}
We seek intersections of the form $\Delta_{\alpha,\sigma} \cap \Delta_{\beta,\tau}$ where $\beta \neq \alpha$. For each $\tau \in K'$, we define the red curves $R_\alpha = \{\Delta_{\alpha,\sigma} : \sigma\neq\alpha\}$ as all curves involving simplex $\alpha$, and the blue curves $B_\tau = \{\Delta_{\beta,\tau} : \beta\neq\tau\}$ as all curves involving simplex $\tau$.
(See Figure~\ref{fig:sector} for the geometric setting; the red and blue curves are
real-valued functions on $S^1$, and we detect their intersection as real-valued crossing
events while sweeping around $S^1$.)

The key insight from pre-refinement is that within the band $(\lambda_1, \lambda_2]$, no two red curves can intersect (these would be intra-$\alpha$ intersections), and no two blue curves can intersect (these would be intra-$\tau$ intersections). 

Each difference curve $\Delta_{\alpha,\sigma}(\omega)$ consists of $O(1)$ trigonometric pieces, switching between different expressions at directions where the active vertex changes. 
Since simplices in $\mathbb{R}^2$ have at most 3 vertices, there are $O(1)$ such transition points. 
When we clip to the band $[\lambda_1, \lambda_2]$, each curve crosses the boundary at $O(1)$ points, subdividing into $O(1)$ monotone arcs within the band.

We apply the standard Bentley-Ottmann sweep-line algorithm to all $O(n)$ arcs. 
Since red arcs do not intersect each other and blue arcs do not intersect each other (by pre-refinement), any intersection must be between a red and blue arc. We terminate immediately upon detecting the first such intersection. This takes $O(n \log n)$ time per choice of $\tau$.
Testing all $O(n)$ choices of $\tau \in K'$ takes $O(n^2 \log n)$ total time for fixed $\alpha$.
\end{proof}

\subsubsection{Predicate 2: Band Refinement}

When the existence test detects an inter-intersection, we must enumerate all candidates involving $\alpha$ to refine the band.

\begin{lemma}
\label{lem:refinement}
When an inter-intersection involving $\alpha$ is detected, we can enumerate all $O(n^3)$ $\alpha$-candidates in $(\lambda_1, \lambda_2]$ and compute a refined band with no $\alpha$-candidates in its interior in $O(n^3 \polylog n)$ time.
\end{lemma}
\begin{proof}
After pre-refinement, only inter-intersections remain in the band's interior. For simplex $\alpha$, there are $O(n^3)$ such intersections (choosing $\sigma, \beta, \tau$ with $O(n)$ choices each). We enumerate these candidates, sort them, and refine the band using binary search with the decision procedure, taking $O(n^3\polylog n)$ time.
\end{proof}

\begin{theorem}
\label{thm:refinement}
After global pre-refinement, processing each simplex $\alpha$ takes $O(n^2 \log n + \mathbbm{1}_\alpha n^3 \polylog n)$ total time, where $\mathbbm{1}_\alpha$ is the indicator function of the existence test.
\end{theorem}

\begin{proof}
This follows directly from Lemma~\ref{lem:existence_test} and~\ref{lem:refinement}.
\end{proof}

\subsubsection{Randomized Analysis}

\begin{theorem}
\label{thm:expected}
Processing simplices from both $K$ and $K'$ in uniformly random order yields $\mathbb{E}[C] = O(\log n)$ refinements, where $C$ counts positive existence tests.
\end{theorem}

\begin{proof}
For each simplex $\alpha$, define $x_\alpha$ as the smallest $\alpha$-candidate that is $\geq \dmax$ and $y_\alpha$ as the largest $\alpha$-candidate that is $< \dmax$. Since $\dmax$ is realized by at least one candidate, for some simplex $\alpha^*$ we have $x_{\alpha^*} = \dmax$. 
Every interval $(y_\alpha, x_\alpha]$ contains $\dmax$, so our maintained band always has the form $(y_{\alpha_i}, x_{\alpha_j}]$ for some processed simplices $\alpha_i, \alpha_j$.

When processing simplices in random order, we obtain random permutations of the sequences $\{x_\alpha\}$ and $\{y_\alpha\}$. 
Simplex $\alpha$ passes the existence test and the refinement is triggered if and only if there exists an $\alpha$-candidate in the current band.
This occurs when $x_\alpha$ is smaller than all previously seen $x_{\alpha'}$ values (making it a left-to-right minimum), or $y_\alpha$ is larger than all previously seen $y_{\alpha'}$ values (making it a left-to-right maximum).

Indeed, if $x_\alpha$ is not a minimum, then some previously processed $\alpha'$ satisfies $x_{\alpha'} \le x_\alpha$, so the band's upper bound is already at most $x_\alpha$, excluding all $\alpha$-candidates with value $\ge x_\alpha$.  
Likewise, if $y_\alpha$ is not a maximum, the band's lower bound already exceeds $y_\alpha$, excluding all $\alpha$-candidates with value $\le y_\alpha$.  
Thus, when neither $x_\alpha$ nor $y_\alpha$ improves the band boundaries, every $\alpha$-candidate lies outside $(y_\alpha, x_\alpha)$, and none can fall in the current band.

By a classical result on random permutations~\cite{Knuth1998}: in a uniformly random ordering of $n$ distinct values, the $i$-th element is a left-to-right minimum with probability $1/i$ (being the smallest of the first $i$ elements). By linearity of expectation, the expected number of left-to-right minima is $H_n = \sum_{i=1}^n \frac{1}{i} = O(\log n)$. The same bound holds for maxima. 
Therefore, the expected number of refinements is at most the expected number of minima in $\{x_\alpha\}$ plus the expected number of maxima in $\{y_\alpha\}$, giving $\mathbb{E}[C] \leq O(\log n) + O(\log n) = O(\log n)$.
\end{proof}

\begin{corollary}[Total complexity in 2D]
\label{cor:total}
The candidate pruning algorithm runs in expected time
\[
O(n \cdot n^2\log n + \log n \cdot n^3\polylog(n)) = \tilde{O}(n^3).
\]
\end{corollary}

\begin{proof}
We process $O(n)$ simplices from both complexes. Each requires an $O(n^2\log n)$ existence test. 
In expectation, only $O(\log n)$ trigger the $O(n^3\polylog(n))$ band refinement. 
The refinement cost dominates, giving $\tilde{O}(n^3)$ expected time.
\end{proof}

\section{Computing $\dint$ in $\R^2$}
\label{Sec:2Dint}
Our framework also improves the computation of the integral bottleneck distance $d_1$
from the previous $\tilde{O}(n^6)$ bound~\cite{munch2025kinetic} to $\tilde{O}(n^5)$.

The integral requires maintaining the bottleneck matching as $\omega$ sweeps around
$\mathbb{S}^1$.
As $\omega$ varies, the bottleneck distance $H(\omega)$ changes only at the $O(n^4)$
difference curve intersection events identified in Section~\ref{Sec:2D}.
At each event, we update the matching in $O(n\log n)$ time via an augmenting path search to
identify which difference curve realizes the new bottleneck.
Between consecutive events at angles $\omega_i$ and $\omega_{i+1}$, the bottleneck distance
is given by a single difference curve that can be integrated analytically.

Thus the integral distance can be computed in $\tilde{O}(n^5)$ total time: $O(n^4)$ events,
each requiring $O(n\log n)$ processing.
This improves upon the hourglass data structure approach of~\cite{munch2025kinetic}, which
maintained two kinetic heaps and required $\tilde{O}(n^6)$ time.

\section{Computing $\dmax$ in $\R^3$}
\label{Sec:3D}
We now consider embedded complexes in $\mathbb{R}^3$. 
Directions range over the sphere $\mathbb{S}^2$, difference curves become surfaces, and intersections occur along curves rather than at points. 
The algorithm follows the same framework as 2D, but requires traversing an arrangement on $\mathbb{S}^2$ via an Euler tour rather than sweeping a circle. 
This yields expected complexity $\tilde{O}(n^5)$.

\subsection{Events}
Let $K, K' \subset \mathbb{R}^3$ be finite simplicial complexes with at most $n$ simplices each. 
For each simplex $\alpha \in K \cup K'$, $\omega \in S^2$, the \emph{insertion surface} is:
\[
I_\alpha : S^2 \to \mathbb{R}, \quad I_\alpha(\omega) := \max_{v \in \alpha} \langle v, \omega \rangle
\]

Each $I_\alpha$ is the upper envelope of at most 4 trigonometric forms on $S^2$. 
The loci where the maximizing vertex changes are great circles $\{\omega \in S^2 : \langle v_i - v_j, \omega \rangle = 0\}$.

For distinct simplices $\alpha, \sigma$, the \emph{difference surface} is a
real-valued function on $\mathbb{S}^2$:
\[
\Delta_{\alpha,\sigma}(\omega) := \begin{cases}
|I_\alpha(\omega) - I_\sigma(\omega)| & \text{if } \alpha, \sigma \text{ are from different complexes} \\
\frac{1}{2}|I_\alpha(\omega) - I_\sigma(\omega)| & \text{if } \alpha, \sigma \text{ are from the same complex}
\end{cases}
\]

Within each region where the active vertices are fixed, $\Delta_{\alpha,\sigma}$ simplifies to $|\langle u - w, \omega \rangle|$ for active vertices $u \in \alpha$ and $w \in \sigma$. 
Two difference surfaces $\Delta_{\alpha,\sigma}$ and $\Delta_{\beta,\tau}$ intersect along curves on $\mathbb{S}^2$ where $\Delta_{\alpha,\sigma}(\omega) = \Delta_{\beta,\tau}(\omega)$.

\begin{lemma}
\label{lem:3d-well-behaved}
For simplicial complexes in $\mathbb{R}^3$:
\begin{enumerate}[noitemsep]
\item Each pair of insertion surfaces intersects along $O(1)$ great circle arcs on $\mathbb{S}^2$.
\item Each pair of difference surfaces intersects along $O(1)$ great circle arcs on $\mathbb{S}^2$.
\item Each intersection curve between difference surfaces has $O(1)$ local maxima.
\item Each difference surface crosses any threshold $\lambda$ along $O(1)$ curves on $\sS^2$.
\end{enumerate}
\end{lemma}

The proof follows from the piecewise structure: since simplices in $\mathbb{R}^3$ have at most 4 vertices, each insertion surface has $O(1)$ regions of fixed active vertex. 
Within each region, the surfaces are trigonometric forms and their intersections are great circles. 
The detailed geometric analysis is in Appendix~\ref{app:3d-geometry}.

\begin{theorem}[Events for $\dmax$ in $\mathbb{R}^3$]
The direction $\omega^*$ attaining $\dmax$ occurs at one of:
\begin{enumerate}[noitemsep]
\item A local maximum of a difference surface $\Delta_{\alpha,\sigma}$
\item A local maximum along an intersection curve on $\mathbb{S}^2$ where two difference surfaces meet
\item A triple intersection point where three difference surfaces meet
\end{enumerate}
\end{theorem}

\begin{proof}
The bottleneck distance function $H(\omega)$ is continuous on $S^2$. At any $\omega^* \in S^2$ where $d_\infty$ is attained, consider the difference surfaces achieving the maximum value $H(\omega^*)$.

If a unique surface $\Delta_{\alpha,\beta}$ achieves this maximum in a neighborhood of $\omega^*$, then by continuity, $\omega^*$ must be a local maximum of that surface (case 1). This covers both matched pairs and unmatched persistence.

If exactly two surfaces achieve the maximum at $\omega^*$ but no other surface does so nearby, then $\omega^*$ lies on their intersection curve on $\mathbb{S}^2$. Since $H$ is locally determined by these two surfaces along the curve, $\omega^*$ must be a local maximum along this curve (case 2).

If three or more surfaces achieve the maximum at $\omega^*$, then generically three surfaces meet at isolated points, giving case 3.
\end{proof}

\begin{corollary}
There are $O(n^6)$ candidate values where $\dmax$ may be attained.
\end{corollary}

\begin{proof}
We count candidates by type. 
Type 1 has $O(n^2)$ difference surfaces (including both between-complex and within-complex pairs), each having $O(1)$ maxima, contributing $O(n^2)$ candidates. 
Type 2 consists of $\binom{O(n^2)}{2} = O(n^4)$ pairs with $O(1)$ maxima per intersection curve, giving $O(n^4)$ candidates. 
Type 3 is $\binom{O(n^2)}{3} = O(n^6)$ triple intersection points. 
This gives us the total of $O(n^6)$ candidates.
\end{proof}

\subsection{Overview of Algorithm}

\paragraph*{Decision procedure.}Given threshold $\lambda$, we decide whether $\dmax \leq \lambda$ by traversing an arrangement on $\sS^2$. 
Similar to the 2D case, we track two types of events:
\begin{itemize}[noitemsep, topsep=0pt]
    \item Diagram changes: $O(n^2)$ intersection curves between insertion surfaces where the persistence diagrams change combinatorially;
    \item Threshold crossings: $O(n^2)$ curves on $\mathbb{S}^2$ where difference surfaces cross the threshold $\lambda$.
\end{itemize}
Overlaying these curves yields an arrangement of $O(n^4)$ complexity on $\sS^2$. 
We traverse the dual of this arrangement via an Euler tour by traversing each edge twice, maintaining a bottleneck matching throughout. 
Each diagram change requires $O(n)$ time to update the barcode pairing using the vineyard algorithm, while each threshold crossing may require $O(n \log n)$ time for augmenting path search.
The total complexity is $\tilde{O}(n^5)$.

\paragraph*{Improved algorithm via band refinement.}
We maintain a band $(\lambda_1, \lambda_2]$ containing $\dmax$ and process simplices in random order. For each simplex $\alpha$:

\begin{enumerate}[noitemsep, topsep=0pt]
\item Existence predicate: Test if any $\alpha$-candidate lies in the band. 
Following the approach from 2D adapted to surfaces: for fixed $\alpha$ and $\sigma \in K'$, consider the surface $\Delta_{\alpha,\sigma}$. 
For each $\beta \in K \setminus \{\alpha\}$, compute intersection curves with all $\Delta_{\beta,\tau}$ surfaces. 
Within the band, these $O(n)$ curves form disjoint arcs on $\mathbb{S}^2$ (no intra-$\beta$ intersections after pre-refinement). 
Similarly get $O(n)$ disjoint arcs from surfaces involving $\gamma \in K \setminus \{\alpha,\beta\}$. 
Check for red-blue intersections in $O(n \log n)$ time. 
Testing all combinations takes $O(n^4 \log n)$.

\item Band refinement: If positive, enumerate all $\alpha$-candidates in the band: $O(n^5)$ triple intersections, $O(n^3)$ curve maxima, $O(n)$ surface maxima.
Then refine via binary search in $\tilde{O}(n^5)$ time.
\end{enumerate}

Using the same randomization analysis as 2D, we expect $O(\log n)$ refinements. 
This results in total expected complexity of $O(n \cdot n^4 \log n + \log n \cdot n^5 \polylog n) = \tilde{O}(n^5)$.

\subsection{Decision Problem}

Given threshold $\lambda \geq 0$, we decide whether $\dmax \leq \lambda$ by constructing and traversing an arrangement on $\sS^2$. 
Unlike the 2D case where we sweep around a circle, here we must systematically explore the entire sphere.

We first construct an arrangement $\mathcal{A}$ on $\sS^2$ consisting of:
\begin{itemize}[noitemsep, topsep=0pt]
\item $O(n^2)$ curves from insertion surface intersections (where persistence diagrams change)
\item $O(n^2)$ curves where difference surfaces equal $\lambda$ (threshold crossings)
\end{itemize}

Each insertion surface is piecewise trigonometric with $O(1)$ pieces, so pairs intersect along $O(1)$ curves each. 
Similarly, each difference surface crosses the threshold $\lambda$ along $O(1)$ curves. 
The overlay of all these curves forms an arrangement with $O(n^4)$ vertices, edges, and faces.

\begin{theorem}
Whether $\dmax \leq \lambda$ can be decided in $\tilde{O}(n^5)$ time.
\end{theorem}

\begin{proof}
We construct the arrangement $\mathcal{A}$ and compute an Euler tour of its dual graph, treated as a planar graph embedded on $\sS^2$. 
The tour visits each edge twice, once in each direction.

Starting from an arbitrary face with a perfect matching in $G_\lambda(\omega)$, we traverse the tour and maintain the matching throughout. 
At insertion curve crossings where the persistence diagram changes, we update the barcode pairing via a single transposition in $O(n)$ time using the vineyard algorithm. 
At threshold crossings where an edge weight exceeds $\lambda$, we search for an augmenting path in $O(n \log n)$ time. 
If no augmenting path exists, we conclude $\dmax > \lambda$ and terminate.
If the tour completes with valid matchings maintained throughout, then $d_B(\omega) \leq \lambda$ for all $\omega \in \sS^2$, hence $\dmax \leq \lambda$.

The arrangement has $O(n^4)$ edges in total. 
The $O(n^4)$ diagram changes contribute $O(n^5)$ total time at $O(n)$ each, while the $O(n^2)$ threshold crossings contribute $\tilde{O}(n^3)$ at $O(n \log n)$ each. 
However, the bottleneck arises from traversing all $O(n^4)$ edges of the arrangement, where at each edge crossing we may need to update the matching. 
Since each edge crossing potentially requires $O(n \log n)$ time for augmenting path search, the total complexity is $O(n^4 \cdot n \log n) = \tilde{O}(n^5)$.
\end{proof}

\subsection{Candidate Pruning via Band Refinement and Complexity Analysis}

As in the 2D case, we maintain a band $(\lambda_1, \lambda_2]$ containing $\dmax$ and process simplices in random order. 
We first apply a pre-refinement to exclude all intra-events from the band.

\paragraph*{Pre-refinement.}
Before processing individual simplices, we refine the band to exclude all intra-events from its interior. 
We call a triple intersection point $\Delta_{\alpha,\sigma} \cap \Delta_{\beta,\tau} \cap \Delta_{\gamma,\rho}$ an \emph{intra-event} if any two of the six involved simplices are equal. 
Let $S$ be the set containing all $O(n^2)$ local maxima of difference surfaces, all $O(n^4)$ local maxima along surface intersection curves, and all intra-event triple intersections. Since there are $O(n^5)$ ways to choose 5 distinct simplices and $O(1)$ ways to make two of them equal, we have $|S| = O(n^5)$. 
We use binary search with the decision procedure to find the largest value in $S$ that is
$< d_\infty$ (setting $\lambda_1$) and the smallest value in $S$ that is $\geq d_\infty$
(setting $\lambda_2$), yielding the band $(\lambda_1, \lambda_2]$ that contains $d_\infty$
but excludes all values of $S$ from its interior.
This takes $O(n^5\polylog n)$ time.

\begin{lemma}
\label{lem:existence-3d}
For fixed $\alpha \in K$, we can test whether any inter-intersection involving $\alpha$ lies in the band $(\lambda_1, \lambda_2]$ in $O(n^4 \log n)$ time.
\end{lemma}
\begin{proof}
Fix $\alpha \in K$ and $\sigma \in K'$, giving the difference surface $\Delta_{\alpha,\sigma}$. 
We seek triple intersection points of the form $\Delta_{\alpha,\sigma} \cap \Delta_{\beta,\tau} \cap \Delta_{\gamma,\rho}$ where $\beta, \gamma \in K \setminus \{\alpha\}$ and $\tau, \rho \in K'$, with the intersection value lying in $(\lambda_1, \lambda_2]$.

For fixed $\beta \in K \setminus \{\alpha\}$, compute the intersection curves between $\Delta_{\alpha,\sigma}$ and all surfaces $\Delta_{\beta,\tau}$ for $\tau \in K'$. 
Each pair of surfaces intersects along great circle arcs on $\mathbb{S}^2$ where $\Delta_{\alpha,\sigma}(\omega) = \Delta_{\beta,\tau}(\omega)$. 
We obtain $O(n)$ intersection curves on $\mathbb{S}^2$.

The crucial observation is that when we restrict to the band $\{\omega : \lambda_1 < \Delta_{\alpha,\sigma}(\omega) \leq \lambda_2\}$, these $O(n)$ curves become disjoint arcs. 
Any intersection between curves $\Delta_{\alpha,\sigma} \cap \Delta_{\beta,\tau_1}$ and $\Delta_{\alpha,\sigma} \cap \Delta_{\beta,\tau_2}$ would yield a triple point where $\Delta_{\alpha,\sigma} = \Delta_{\beta,\tau_1} = \Delta_{\beta,\tau_2}$. 
This is an intra-event (two surfaces share $\beta$), excluded from the band's interior via pre-refinement.

Similarly, for fixed $\gamma \in K \setminus \{\alpha, \beta\}$, the intersection curves between $\Delta_{\alpha,\sigma}$ and all $\Delta_{\gamma,\rho}$ surfaces yield another set of $O(n)$ disjoint arcs within the band. 
We now have two collections of arcs on the restricted domain of $\Delta_{\alpha,\sigma}$: the ``red'' arcs from $\beta$-surfaces and the ``blue'' arcs from $\gamma$-surfaces.
A red-blue intersection corresponds to a triple point $\Delta_{\alpha,\sigma} \cap \Delta_{\beta,\tau} \cap \Delta_{\gamma,\rho}$, which is an inter-intersection involving $\alpha$.

To detect red-blue intersections, we apply stereographic projection from the north pole of $\mathbb{S}^2$ to map the band region to a planar domain and use a sweep-line algorithm. 
As shown in Appendix~\ref{app:stereographic}, the projection preserves the disjointness of red and blue arc families while mapping great circle arcs to well-behaved curves (circles and lines) that intersect $O(1)$ times.
The sweep-line algorithm thus detects all red-blue intersections in $O(n \log n)$ time.

Iterating over all choices of $\sigma \in K'$, $\beta \in K \setminus \{\alpha\}$, and $\gamma \in K \setminus \{\alpha,\beta\}$ gives total complexity $O(n \cdot n \cdot n \cdot n \log n) = O(n^4 \log n)$.
\end{proof}

\begin{lemma}
\label{lem:enumeration-3d}
When the existence test is positive for simplex $\alpha$, we can enumerate all $O(n^5)$ candidates involving $\alpha$ that lie in $(\lambda_1, \lambda_2]$.
\end{lemma}

\begin{proof}
For each $\alpha$, the pre-refinement step excludes the $O(n)$ local maxima of difference surfaces and $O(n^2)$ local maxima of intersection curves.
We only need to consider triple intersection points $\Delta_{\alpha,\sigma} \cap \Delta_{\beta,\tau} \cap \Delta_{\gamma,\rho}$: With $O(n)$ choices each for $\sigma, \beta, \tau, \gamma, \rho$, and $O(1)$ intersection points per triple, we get $O(n^5)$ candidates.
Total enumeration takes $O(n^5)$ time. Binary search with the decision procedure then refines the band in $\tilde{O}(n^5)$ time.
\end{proof}

\begin{theorem}[Total complexity in 3D]
\label{thm:complexity-3d}
The algorithm runs in expected time $\tilde{O}(n^5)$.
\end{theorem}

\begin{proof}
We process $O(n)$ simplices from both complexes. For each simplex $\alpha$, we perform the existence test in $O(n^4 \log n)$ time. 
If positive, we enumerate and refine in $\tilde{O}(n^5)$ time.
We expect $O(\log n)$ refinements using the same analysis as in 2D.
The total expected complexity is therefore $O(n \cdot n^4 \log n + \log n \cdot n^5 \polylog n) = \tilde{O}(n^5)$.
\end{proof}

\section{Conclusion}
\label{Sec:Conclusion}
We have presented the first algorithms for computing the maximum bottleneck distance between Persistent Homology Transforms, achieving expected time complexity of $\tilde{O}(n^3)$ in $\mathbb{R}^2$ and $\tilde{O}(n^5)$ in $\mathbb{R}^3$. As a byproduct, our framework also improves the integral bottleneck distance computation from $\tilde{O}(n^6)$ to $\tilde{O}(n^5)$ in 2D. Using similar techniques, the integral version in $\mathbb{R}^3$ can be computed in $\tilde{O}(n^7)$ time by maintaining matchings across all $O(n^6)$ events, though we leave the technical details for future work.

Our adaptation of the framework introduced by Bjerkevik and Kerber~\cite{Bjerkevik2023} reveals a broader principle: directional transforms that vary continuously over parameter spaces can often be optimized through event-driven algorithms and randomized candidate pruning. This suggests potential applications beyond the PHT to other parametric problems in computational topology, such as the Euler Characteristic Transform or persistence-based shape signatures.

From a practical perspective, the cubic complexity in 2D makes exact PHT comparison feasible for moderate-sized datasets, potentially enabling new applications in shape matching and retrieval where sampling-based approximations were previously necessary. However, significant challenges remain: extending to the Wasserstein distance would require fundamentally different techniques due to its non-combinatorial nature, and handling higher-dimensional ambient spaces faces steep complexity barriers.

\bibliographystyle{plain}
\bibliography{maxBottleBib}

\begin{thebibliography}{10}

\bibitem{Bjerkevik2023}
Havard~Bakke Bjerkevik and Michael Kerber.
\newblock Asymptotic improvements on the exact matching distance for $2$-parameter persistence.
\newblock {\em Journal of Computational Geometry}, 14(1):309–342, Dec. 2023.

\bibitem{Cabello2024Matching}
Sergio Cabello, Siu-Wing Cheng, Otfried Cheong, and Christian Knauer.
\newblock {Geometric Matching and Bottleneck Problems}.
\newblock In Wolfgang Mulzer and Jeff~M. Phillips, editors, {\em 40th International Symposium on Computational Geometry (SoCG 2024)}, volume 293 of {\em Leibniz International Proceedings in Informatics (LIPIcs)}, pages 31:1--31:15, Dagstuhl, Germany, 2024. Schloss Dagstuhl -- Leibniz-Zentrum f{\"u}r Informatik.

\bibitem{carriere2020perslay}
Mathieu Carrière, Frederic Chazal, Yuichi Ike, Theo Lacombe, Martin Royer, and Yuhei Umeda.
\newblock Perslay: A neural network layer for persistence diagrams and new graph topological signatures.
\newblock In Silvia Chiappa and Roberto Calandra, editors, {\em Proceedings of the Twenty Third International Conference on Artificial Intelligence and Statistics}, volume 108 of {\em Proceedings of Machine Learning Research}, pages 2786--2796. PMLR, 26--28 Aug 2020.

\bibitem{CohenSteiner2006Vineyard}
David Cohen-Steiner, Herbert Edelsbrunner, and Dmitriy Morozov.
\newblock Vines and vineyards by updating persistence in linear time.
\newblock In {\em Proceedings of the Twenty-Second Annual Symposium on Computational Geometry}, SCG '06, page 119–126, New York, NY, USA, 2006. Association for Computing Machinery.

\bibitem{CurryPHT2018}
Justin Curry, Sayan Mukherjee, and Katharine Turner.
\newblock How many directions determine a shape and other sufficiency results for two topological transforms.
\newblock {\em Transactions of the American Mathematical Society, Series B}, 2018.

\bibitem{Dey2018Bottleneck}
Tamal~K. Dey and Cheng Xin.
\newblock {Computing Bottleneck Distance for 2-D Interval Decomposable Modules}.
\newblock In Bettina Speckmann and Csaba~D. T\'{o}th, editors, {\em 34th International Symposium on Computational Geometry (SoCG 2018)}, volume~99 of {\em Leibniz International Proceedings in Informatics (LIPIcs)}, pages 32:1--32:15, Dagstuhl, Germany, 2018. Schloss Dagstuhl -- Leibniz-Zentrum f{\"u}r Informatik.

\bibitem{DeyWang2017}
Tamal~Krishna Dey and Yusu Wang.
\newblock {\em Computational {{Topology}} for {{Data Analysis}}}.
\newblock Cambridge University Press, 1 edition, 2017.

\bibitem{EdelsHarer2010}
Herbert Edelsbrunner and John Harer.
\newblock {\em Computational Topology - an Introduction.}
\newblock American Mathematical Society, 2010.

\bibitem{edelsbrunner2024mergetree}
Herbert Edelsbrunner and Teresa Heiss.
\newblock Merge trees of periodic filtrations, 2024.

\bibitem{Efrat2001}
Alon Efrat, Alon Itai, and Matthew~J.\ Katz.
\newblock Geometry helps in bottleneck matching and related problems.
\newblock {\em Algorithmica}, 31(1):1--28, Sep 2001.

\bibitem{GhristPHT2018}
Robert Ghrist, Rachel Levanger, and Huy Mai.
\newblock Persistent homology and {Euler} integral transforms.
\newblock {\em Journal of Applied and Computational Topology}, 2(1):55--60, Oct 2018.

\bibitem{giusti2015clique}
Chad Giusti, Eva Pastalkova, Carina Curto, and Vladimir Itskov.
\newblock Clique topology reveals intrinsic geometric structure in neural correlations.
\newblock {\em Proceedings of the National Academy of Sciences}, 112(44):13455--13460, 2015.

\bibitem{Hensel2021Survey}
Felix Hensel, Michael Moor, and Bastian Rieck.
\newblock A survey of topological machine learning methods.
\newblock {\em Frontiers in Artificial Intelligence}, Volume 4 - 2021, 2021.

\bibitem{hiraoka2016hierarchical}
Yasuaki Hiraoka, Takenobu Nakamura, Akihiko Hirata, Emerson~G. Escolar, Kaname Matsue, and Yasumasa Nishiura.
\newblock Hierarchical structures of amorphous solids characterized by persistent homology.
\newblock {\em Proceedings of the National Academy of Sciences}, 113(26):7035--7040, 2016.

\bibitem{HopcroftKarp1973}
John~E. Hopcroft and Richard~M. Karp.
\newblock An $n^{5/2} $ algorithm for maximum matchings in bipartite graphs.
\newblock {\em SIAM J. Comput.}, 2(4):225–231, Dec 1973.

\bibitem{KerberGeometryHelps2017}
Michael Kerber, Dmitriy Morozov, and Arnur Nigmetov.
\newblock Geometry helps to compare persistence diagrams.
\newblock {\em ACM J. Exp. Algorithmics}, 22, sep 2017.

\bibitem{Knuth1998}
Donald~E. Knuth.
\newblock {\em The Art of Computer Programming, Volume 3: Sorting and Searching}.
\newblock Addison-Wesley, 2nd edition, 1998.
\newblock Section 5.1.1.

\bibitem{munch2025kinetic}
Elizabeth Munch, Elena~Xinyi Wang, and Carola Wenk.
\newblock The kinetic hourglass data structure for computing the bottleneck distance of dynamic data, 2025.

\bibitem{reimann2017cliques}
Michael~W. Reimann, Max Nolte, Martina Scolamiero, Katharine Turner, Rodrigo Perin, Giuseppe Chindemi, Paweł Dłotko, Ran Levi, Kathryn Hess, and Henry Markram.
\newblock Cliques of neurons bound into cavities provide a missing link between structure and function.
\newblock {\em Frontiers in Computational Neuroscience}, Volume 11 - 2017, 2017.

\bibitem{TurnerPHT2014}
Katharine Turner, Sayan Mukherjee, and Doug~M. Boyer.
\newblock {Persistent homology transform for modeling shapes and surfaces}.
\newblock {\em Information and Inference: A Journal of the IMA}, 3(4):310--344, 12 2014.

\end{thebibliography}

\appendix
\section{Properties of difference curves}
\subsection{Local maxima}
\label{app:maxima}
\begin{lemma}
\label{lem:diff_curve_maxima}
Let $K, K' \subset \mathbb{R}^2$ be simplicial complexes, and let $\Delta_{\alpha,\sigma}(\omega)$ be a difference curve for simplices $\alpha \in K$ and $\sigma \in K'$ with $\omega \in \mathbb{S}^1$. Then $\Delta_{\alpha,\sigma}$ has $O(1)$ local maxima on $\mathbb{S}^1$, which can be computed by solving algebraic equations.
\end{lemma}

\begin{proof}
Recall that for a simplex $\alpha$ with vertices $\{v_0, \ldots, v_k\}$, the insertion curve is
\[
I_\alpha(\omega) = \max_{0 \leq i \leq k} \langle v_i, \omega \rangle.
\]
This is the upper envelope of at most $k+1$ trigonometric functions on $\mathbb{S}^1$. 
Since each pair of trigonometric functions $\langle v_i, \omega \rangle$ and $\langle v_j, \omega \rangle$ can be equal for at most two directions $\omega \in \mathbb{S}^1$, the function $I_\alpha$ is piecewise trigonometric with $O(k)$ pieces.

The difference curve $\Delta_{\alpha,\sigma}(\omega) = |I_\alpha(\omega) - I_\sigma(\omega)|$ is thus piecewise differentiable on $\mathbb{S}^1$. 
Local maxima can occur only at points where $\Delta_{\alpha,\sigma}$ is non-differentiable or where its derivative vanishes. 
The non-differentiable points include directions where $I_\alpha(\omega) = I_\sigma(\omega)$ (where the difference curve has a cusp at zero) and directions where the active vertex changes in either $I_\alpha$ or $I_\sigma$ (corners in the insertion curves). 
At differentiable points, if $I_\alpha(\omega) > I_\sigma(\omega)$ locally, then $\Delta_{\alpha,\sigma} = I_\alpha - I_\sigma$, and critical points occur where the directional derivatives of the active height functions are equal.

For simplicial complexes in $\mathbb{R}^2$, we have $|\alpha| \leq 3$ and $|\sigma| \leq 3$ (triangles at most). 
The insertion curve $I_\alpha$ has $O(1)$ pieces, and similarly for $I_\sigma$. 
There are $O(1)$ directions where the active vertex changes, the equation $I_\alpha(\omega) = I_\sigma(\omega)$ yields $O(1)$ solutions, and within each region of fixed active vertices, there are $O(1)$ critical points. 
Thus $\Delta_{\alpha,\sigma}$ has $O(1)$ candidate points for local maxima.

To compute these points explicitly, parameterize $\omega \in \mathbb{S}^1$ as $\omega = (\cos\theta, \sin\theta)$. 
Then $\langle v_i, \omega \rangle = x_i\cos\theta + y_i\sin\theta$ where $v_i = (x_i, y_i)$. 
The conditions for critical points become trigonometric equations: active vertex changes occur when $x_i\cos\theta + y_i\sin\theta = x_j\cos\theta + y_j\sin\theta$, and insertion curves meet when $\max_i(x_i\cos\theta + y_i\sin\theta) = \max_j(u_j\cos\theta + v_j\sin\theta)$ where $(u_j, v_j)$ are vertices of $\sigma$. 
These reduce to trigonometric equations with finitely many solutions in $[0, 2\pi)$, each solvable in closed form.
\end{proof}

\subsection{Proof of Lemma~\ref{lem:event-bound}}
\label{appendix:event-bound-proof}
\begin{proof}
We first establish that each difference curve $\Delta_{\alpha,\sigma}$ crosses any horizontal line $y = \lambda$ at most $O(1)$ times as $\omega$ varies around $S^1$.

For any pair $(\alpha, \sigma) \in K \times K'$, the difference curve $\Delta_{\alpha,\sigma}(\omega) = |I_\alpha(\omega) - I_\sigma(\omega)|$ is piecewise trigonometric. 
Since simplices in $\mathbb{R}^2$ have at most 3 vertices, there are at most 12 sectors on $S^1$ where the active vertices of both $\alpha$ and $\sigma$ remain fixed.

We call a maximal interval of directions where the active vertices remain fixed a \emph{sector}. 
Within each sector with active vertices $u \in \alpha$ and $w \in \sigma$, we have 
$$\Delta_{\alpha,\sigma}(\omega) = |\langle u - w, \omega \rangle|.$$
Parameterizing $\omega = (\cos\theta, \sin\theta)$ and letting $u - w = (d_x, d_y)$, this becomes
$$\Delta_{\alpha,\sigma}(\theta) = |d_x \cos\theta + d_y \sin\theta| = r|\cos(\theta - \phi)|$$
where $r = \|u - w\|$ and $\phi = \arctan(d_y/d_x)$.

The equation $r|\cos(\theta - \phi)| = \lambda$ has solutions when $|\cos(\theta - \phi)| = \lambda/r$. For $\lambda < r$, this yields exactly 4 solutions in $[0, 2\pi)$:
$$\theta \in \{\phi \pm \arccos(\lambda/r), \phi + \pi \pm \arccos(\lambda/r)\}.$$
For $\lambda = r$, there are 2 solutions, and for $\lambda > r$, there are no solutions.

Since each sector contributes at most 4 crossings and there are at most 12 sectors, each difference curve crosses any horizontal line at most $O(1)$ times.
With $O(n^2)$ difference curves total (one for each pair $(\alpha, \sigma) \in K \times K'$), we obtain $O(n^2)$ crossings with any horizontal line.

For insertion-curve intersections, when two insertion curves $I_\alpha$ and $I_\beta$ from the same complex meet at direction $\omega$, we have $I_\alpha(\omega) = I_\beta(\omega)$. This forces all difference curves $\Delta_{\alpha,\sigma}(\omega) = \Delta_{\beta,\sigma}(\omega)$ for any $\sigma$ to equal zero at that direction. Thus insertion-curve intersections correspond precisely to difference curves crossing the line $y = 0$. With $O(n)$ insertion curves from each complex, there are at most $\binom{O(n)}{2} = O(n^2)$ such intersections.
\end{proof}

\section{Treatment of Overlapping Difference Curves}
\label{app:overlaps}

\subsection{The Overlap Issue}
When simplices share vertices, their insertion and difference curves/surfaces may coincide over certain regions. 
Specifically, if simplices $\alpha, \beta \in K$ both contain vertex $v$, then in any direction $\omega$ where $v$ is the active vertex for both simplices, we have $I_\alpha(\omega) = I_\beta(\omega) = \langle v, \omega \rangle$.

This creates two types of overlaps:
\begin{enumerate}[noitemsep, topsep=0pt]
\item Insertion overlaps: When $I_\alpha = I_\beta$ over some region
\item Difference overlaps: When $\Delta_{\alpha,\sigma} = \Delta_{\beta,\tau}$ over some region due to shared active vertices
\end{enumerate}

\subsection{Why Overlaps Don't Affect the Algorithm}
Despite these geometric coincidences, our algorithm remains correct because:

\begin{lemma}
Each difference curve $\Delta_{\alpha,\sigma}$ maintains a unique combinatorial identity determined by the ordered pair $(\alpha, \sigma)$, independent of geometric overlaps.
\end{lemma}

\begin{proof}
In the bottleneck matching, each edge in the bipartite graph corresponds to a specific pair of simplices, not just their geometric values. 
When computing the perfect matching at direction $\omega$, an edge $(p,q)$ in the graph is determined by:
\begin{itemize}[noitemsep, topsep=0pt]
\item The birth-death pair $(b_\alpha, d_\alpha)$ from simplex $\alpha$'s contribution to $\text{Dgm}(h^K_\omega)$
\item The birth-death pair $(b_\sigma, d_\sigma)$ from simplex $\sigma$'s contribution to $\text{Dgm}(h^{K'}_\omega)$
\end{itemize}
Even if $\Delta_{\alpha,\sigma}(\omega) = \Delta_{\beta,\tau}(\omega)$ geometrically, these represent distinct edges in the matching.
\end{proof}

\subsection{Candidate Values Under Overlaps}
The overlapping of difference curves raises the question: are our identified candidates still sufficient to capture $d_\infty$?

\begin{theorem}
Even with overlapping difference curves, the maximum bottleneck distance $d_\infty$ is attained at one of our identified candidate values.
\end{theorem}

\begin{proof}
Let $\omega^*$ be a direction where $d_\infty$ is attained. At $\omega^*$, at least one difference curve achieves the value $d_\infty$. We consider all possible configurations.

If a single combinatorially distinct curve $\Delta_{\alpha,\sigma}$ achieves the maximum value $d_\infty$ in a neighborhood of $\omega^*$, then by Lemma~\ref{lem:continuity} (Lipschitz continuity) $\omega^*$ must be a local maximum of $\Delta_{\alpha,\sigma}$. This is a type-(i) candidate in our enumeration. Note that other curves may coincide with $\Delta_{\alpha,\sigma}$ geometrically at $\omega^*$, but as long as $\Delta_{\alpha,\sigma}$ is locally maximal, it appears in our candidate set.

Suppose multiple curves achieve the value $d_\infty$ at $\omega^*$. Consider any two such curves $\Delta_{\alpha,\sigma}$ and $\Delta_{\beta,\tau}$. If these curves are geometrically distinct in every neighborhood of $\omega^*$, then they cross transversely at $\omega^*$, making it an isolated intersection point. This is a type-(ii) candidate.

The more subtle case occurs when $\Delta_{\alpha,\sigma}$ and $\Delta_{\beta,\tau}$ coincide over an open interval $I$ containing $\omega^*$. This happens when the simplices share vertices that remain active throughout the interval. Over this interval, both curves have the same value; if that value realizes the bottleneck distance throughout $I$, then $H$ is constant on $I$. The maximum over $I$ must be attained somewhere, and there are two possibilities. Either the maximum occurs at a critical point within $I$, in which case this critical point of $\Delta_{\alpha,\sigma}$ (and simultaneously of $\Delta_{\beta,\tau}$) is a type-(i) candidate. Or the maximum occurs at a boundary point of $I$, where the curves begin to separate. Such boundary points are precisely where the active vertices change or where the curves transition from overlapping to non-overlapping, and these are captured as intersection events in our type-(ii) candidates.

In the most general case, many curves may coincide at $\omega^*$, but the same analysis applies. The value $d_\infty$ is realized by at least one combinatorially distinct curve, and either $\omega^*$ is a critical point of this curve or it is a point where curves meet or separate. Both cases are included in our candidate enumeration.

Therefore, even with arbitrary overlaps between difference curves, our candidate set contains all points where $d_\infty$ can be attained.
\end{proof}

\subsection{Counting and Complexity}

Overlaps do not increase our asymptotic complexity bounds:

\begin{lemma}
The number of combinatorially distinct candidates remains $O(n^4)$ in 2D and $O(n^6)$ in 3D, even with overlaps.
\end{lemma}

\begin{proof}
Each candidate corresponds to a combinatorially distinct event:
\begin{itemize}[noitemsep, topsep=0pt]
\item Local maxima: Each of the $O(n^2)$ distinct difference curves contributes $O(1)$ maxima
\item Intersections: Even when curves overlap, the points where they begin or cease to coincide are intersection events. Since we have $O(n^2)$ curves with $O(1)$ pieces each, there are $O(n^4)$ such events in 2D
\item The 3D analysis follows similarly with $O(n^6)$ triple intersection points
\end{itemize}

Geometric coincidences may cause multiple combinatorial events to occur at the same value, but this only affects the constants, not the asymptotic count.
\end{proof}

\section{Geometric Properties in 3D}
\begin{proof}[Proof of Lemma~\ref{lem:3d-well-behaved}]
\label{app:3d-geometry}

For part (1), consider two insertion surfaces $I_\alpha$ and $I_\beta$ where $\alpha$ and $\beta$ are simplices in $\mathbb{R}^3$ with at most 4 vertices each. 
The insertion surface $I_\alpha(\omega) = \max_{v \in \alpha} \langle v, \omega \rangle$ is the upper envelope of at most 4 trigonometric functions on $\mathbb{S}^2$. 

The domain $\mathbb{S}^2$ is partitioned into regions by great circles where pairs of vertices achieve equal heights. 
For vertices $v_i, v_j \in \alpha$, the equation $\langle v_i, \omega \rangle = \langle v_j, \omega \rangle$ yields $\langle v_i - v_j, \omega \rangle = 0$, defining a great circle. 
With $\binom{4}{2} = 6$ pairs of vertices, we get at most 6 great circles partitioning $\mathbb{S}^2$ into $O(1)$ cells. 
Within each cell, $I_\alpha$ has a fixed active vertex.

The intersection $I_\alpha(\omega) = I_\beta(\omega)$ occurs when $\max_{v \in \alpha} \langle v, \omega \rangle = \max_{w \in \beta} \langle w, \omega \rangle$.
Within cells where the active vertices are $u \in \alpha$ and $v \in \beta$, this reduces to $\langle u, \omega \rangle = \langle v, \omega \rangle$, or $\langle u - v, \omega \rangle = 0$, defining a great circle.
Since there are $O(1) \times O(1) = O(1)$ pairs of cells, we get $O(1)$ great circle arcs total.

For part (2), consider difference surfaces $\Delta_{\alpha,\sigma}(\omega) = |I_\alpha(\omega) - I_\sigma(\omega)|$ and $\Delta_{\beta,\tau}(\omega) = |I_\beta(\omega) - I_\tau(\omega)|$. 
Within regions where active vertices are fixed---say $(u_1, w_1)$ for $(\alpha, \sigma)$ and $(u_2, w_2)$ for $(\beta, \tau)$---we have:
\[\Delta_{\alpha,\sigma}(\omega) = |\langle u_1 - w_1, \omega \rangle| = \Delta_{\beta,\tau}(\omega) = |\langle u_2 - w_2, \omega \rangle|\]

This equation $|\langle u_1 - w_1, \omega \rangle| = |\langle u_2 - w_2, \omega \rangle|$ splits into four cases based on signs:
\begin{align}
\langle u_1 - w_1, \omega \rangle &= \langle u_2 - w_2, \omega \rangle \quad \Rightarrow \quad \langle (u_1 - w_1) - (u_2 - w_2), \omega \rangle = 0\\
\langle u_1 - w_1, \omega \rangle &= -\langle u_2 - w_2, \omega \rangle \quad \Rightarrow \quad \langle (u_1 - w_1) + (u_2 - w_2), \omega \rangle = 0
\end{align}
Each case defines a great circle. 
With two sign choices for each surface, we get at most 4 great circles per pair of regions. 
Since each surface has $O(1)$ regions (from $O(1)$ vertex pairs), the total intersection consists of $O(1)$ great circle arcs.

For part (3), along an intersection curve $\mathcal{C}$ where $\Delta_{\alpha,\sigma}(\omega) = \Delta_{\beta,\tau}(\omega)$, local maxima of the common value occur at critical points. 
Parameterizing a great circle arc as $\omega(t) = \cos(t)p + \sin(t)q$ where $p, q$ are orthonormal vectors in $\mathbb{R}^3$, the function along the arc becomes:
\[f(t) = |\langle u - w, \cos(t)p + \sin(t)q \rangle| = |a\cos(t) + b\sin(t)| = \sqrt{a^2 + b^2}|\cos(t - \phi)|\]
where $a = \langle u - w, p \rangle$, $b = \langle u - w, q \rangle$, and $\tan(\phi) = b/a$. This function has exactly 2 local maxima at $t = \phi$ and $t = \phi + \pi$ on $[0, 2\pi)$. 
Additionally, maxima can occur at boundaries where regions change. With $O(1)$ arcs and $O(1)$ boundaries, we get $O(1)$ local maxima total.

For part (4), within a region of fixed active vertices $u \in \alpha$ and $w \in \sigma$, the threshold equation $\Delta_{\alpha,\sigma}(\omega) = \lambda$ becomes $|\langle u - w, \omega \rangle| = \lambda$. Let $d = u - w$ with $\|d\| = r$. 
This splits into:
$\langle d, \omega \rangle = \pm\lambda$.
Each equation defines a plane at distance $\lambda/r$ from the origin. 
The intersection of such a plane with $\mathbb{S}^2$ is a circle (a "small circle" on the sphere) unless $\lambda = r$, in which case it's a point, or $\lambda > r$, in which case it's empty. 
These two circles are parallel, symmetric about the great circle perpendicular to $d$. 
With $O(1)$ regions per difference surface and at most 2 circles per region, we obtain $O(1)$ threshold curves total.
\end{proof}

\subsection{Stereographic Projection for Sweep-Line Algorithm}
\label{app:stereographic}

\begin{lemma}[Stereographic projection preserves structure]
When we apply stereographic projection from the north pole to map a band region on $\sS^2$ to the plane, the resulting curves maintain the properties needed for the sweep-line algorithm:
\begin{enumerate}[noitemsep, topsep=0pt]
\item Great circle arcs map to circles and lines in the plane.
\item Disjoint arc families remain disjoint after projection.
\item Each projected curve can be decomposed into $O(1)$ x-monotone arcs.
\end{enumerate}
\end{lemma}

\begin{proof}
Stereographic projection from the north pole $(0,0,1)$ maps a point $(\theta, \phi)$ on $\sS^2$ to 
\[
(x,y) = \left(\frac{\sin\phi\cos\theta}{1-\cos\phi}, \frac{\sin\phi\sin\theta}{1-\cos\phi}\right)
\]

Great circles on $\sS^2$ are intersections with planes through the origin. Under stereographic projection, these map to circles in the plane (or lines if the plane passes through the north pole). 
Since our band excludes a neighborhood of the poles and each difference surface has $O(1)$ pieces, the projected curves are well-behaved with bounded curvature.

Each projected circle can be decomposed into at most 2 x-monotone arcs (split at the leftmost and rightmost points), while lines are already x-monotone. 
Since each intersection curve consists of $O(1)$ great circle arcs, we obtain $O(1)$ x-monotone pieces per curve after projection.

Within the band $\{\omega : \lambda_1 < \Delta_{\alpha,\sigma}(\omega) \leq \lambda_2\}$, the red arcs (from $\beta$-surfaces) are pairwise disjoint by construction---any intersection would be an intra-$\beta$ event excluded by pre-refinement. 
Similarly, blue arcs (from $\gamma$-surfaces) are pairwise disjoint. 
The projection preserves this disjointness since it is a homeomorphism away from the north pole.
\end{proof}

\end{document}